\documentclass[conference,a4paper]{IEEEtran}

\usepackage{amsmath,amsfonts,amsthm,braket,amsthm,graphicx,subfig}

\theoremstyle{definition} 
\theoremstyle{definition}

\newtheorem {lemma} {Lemma}

\newcommand{\kb}[1]{\ket{#1}\bra{#1}}

\newcommand{\re}[2]{\mathcal{R}_{#1,#2}}

\begin{document}

\title{Asymptotic Analysis of a Three State Quantum Cryptographic Protocol}
\author{
\IEEEauthorblockN{Walter O. Krawec}
\IEEEauthorblockA{Computer Science Department\\
Iona College\\
New Rochelle, NY 10801\\
Email: walter.krawec@gmail.com}
}

\maketitle
\begin{abstract}
In this paper we consider a three-state variant of the BB84 quantum key distribution (QKD) protocol.  We derive a new lower-bound on the key rate of this protocol in the asymptotic scenario and use mismatched measurement outcomes to improve the channel estimation.  Our new key rate bound remains positive up to an error rate of $11\%$, exactly that achieved by the four-state BB84 protocol.%This is also a substantial improvement in the tolerated noise over previous lower-bounds which do not make use of mismatched measurement results.
\end{abstract}


\section{Introduction}

A quantum key distribution (QKD) protocol is designed to allow two users Alice $A$ and Bob $B$ to establish a shared secret key that is secure against an all powerful adversary Eve $E$.  Typically, these protocols operate by $A$ preparing and sending qubits to $B$ who will then measure them; both the preparation and the subsequent measurements are performed in a variety of bases as dictated by the protocol.  This \emph{quantum communication stage} results in $A$ and $B$ each distilling a \emph{raw key}: a string of classical bits of size $N$ which is partially correlated and partially secret.  The users will then run an error correcting (EC) and privacy amplification (PA) protocol resulting in a secret key of size $\ell(N) \le N$.  In this paper, we are interested in the key rate in the asymptotic scenario defined $r = \lim_{N\rightarrow\infty}\frac{\ell(N)}{N}$.  See \cite{QKD-survey} for more information on these standard definitions and processes.

It was shown in \cite{QKD-renner-keyrate} that, assuming collective attacks (where $E$ performs the same attack operation each iteration but is free to postpone the measurement of her ancilla to any future time), then $r = \inf(S(A|E) - H(A|B))$.  Here $S(A|E)$ is the conditional von Neumann entropy; $H(A|B)$ is the conditional classical entropy, and the infimum is over the set of all collective attacks which induce the observed statistics.

In this paper, we will compute a new lower bound on the key rate of a three-state protocol first introduced in \cite{QKD-BB84-three-state,QKD-BB84-three-state-v2}.  Such a protocol allows $A$ to only send $\ket{0}$, $\ket{1}$, or $\ket{+}$ (we will actually consider a generalized version where the third state is $\ket{a} = \alpha\ket{0} + \sqrt{1-\alpha^2}\ket{1}$ for any $\alpha\in(0,1)$); thus the users cannot measure the probability of $E$'s attack flipping a $\ket{-}$ to a $\ket{+}$ as can be done in the four-state BB84 \cite{QKD-BB84} protocol.  However, by using mismatched measurement outcomes \cite{QKD-mismatch1} we can impose further restrictions on the set of possible attacks used by $E$ thus improving the key rate bound (in particular, we will not discard, as is typically done, the measurement outcomes when $A$ and $B$'s choice of basis do not agree).

Our results show that, if we assume $E$'s attack is symmetric (which could even be enforced), then the three state protocol's maximally tolerated error rate is equal to that of the four-state BB84 protocol - i.e., our key rate bound (which will be a function only of parameters that may be observed by $A$ and $B$) remains positive up to an error rate of $11\%$.  We will also consider how the choice of $\ket{a}$ affects this rate.

In \cite{QKD-Tom-threestate1}, mismatched measurement bases were also used to show this three-state protocol's key rate was equal to that of the full four-state BB84.  However, in this paper, we provide an alternative proof of this result using different methods.  The technique we derive here may be easily extended to other QKD protocols.  We also discuss the choice of $\alpha$.

Also, mismatched measurement outcomes were used in \cite{QKD-mismatch1} with the BB84 (a four-state protocol) and the six-state BB84 protocols; it was also shown to produce a superior key rate for certain quantum channels, for those two protocols.  In \cite{QKD-mismatch2} they were used to detect an attacker with greater probability for measure/resend attacks.  We also used them in \cite{SQKD-Krawec-ReflectSecurity} in the proof of security for a semi-quantum QKD protocol.

%%To our knowledge, we are the first to consider the use of mismatched outcomes to the three-state BB84 protocol.

For notation, we denote by $\rho_{AB}$ to mean a density operator acting on the joint Hilbert space $\mathcal{H}_A\otimes\mathcal{H}_B$.  If we write $\rho_{A}$ then we mean the operator resulting from the tracing out of $B$'s subspace (i.e., $\rho_A = tr_B\rho_{AB}$).  These definitions extend to three or more subspaces.  By $S(AB)_\rho$ we mean the von Neumann entropy $S(\rho_{AB})$ and $S(A|B)_\rho$ to mean $S(AB)_\rho - S(B)_\rho$.  We use $H(\cdot)$ to denote the classical Shannon entropy and $h(x)$ to be the binary entropy function (i.e., $h(x) = H(x, 1-x) = -x\log x - (1-x)\log(1-x)$).  All logarithms in this paper are base two.

\section{The Protocol}

The protocol we consider in this paper is a three state one and it is a generalization of the protocol described in \cite{QKD-BB84-three-state,QKD-BB84-three-state-v2}.  Let $\mathcal{B} = \{\ket{0}, \ket{1}\}$ be an arbitrary orthonormal basis and let $\mathcal{A} = \{\ket{a}, \ket{\bar{a}}\}$, where $\ket{a} = \alpha\ket{0} + \sqrt{1-\alpha^2}\ket{1}$, $\ket{\bar{a}} = \sqrt{1-\alpha^2}\ket{0} - \alpha\ket{1}$, and $\alpha \in (0,1)$.  Clearly $\mathcal{A}$ is also an orthonormal basis; note that when $\alpha = 1/\sqrt{2}$, we have $\ket{a} = \ket{+}$ and $\ket{\bar{a}} = \ket{-}$ where $\ket{\pm} = 1/\sqrt{2}(\ket{0} \pm \ket{1})$.  The value of $\alpha$ is considered to be public knowledge.

The quantum communication stage of the protocol consists of the following process:

\begin{enumerate}
\item $A$ prepares a qubit of the form $\ket{0}$, $\ket{1}$, or $\ket{a}$, choosing each with probability $p/2$, $p/2$, and $1-p$ respectively.  This qubit is sent to $B$.  Note that $A$ cannot prepare $\ket{\bar{a}}$.
  \item $B$ chooses with probability $q$ to measure the qubit in the $\mathcal{B} = \{\ket{0}, \ket{1}\}$ basis; otherwise, he measures in the $\mathcal{A} = \{\ket{a}, \ket{\bar{a}}\}$ basis.
  \item $A$ and $B$ will disclose their choice of basis (using an authenticated classical channel).  If their choice of basis is $\mathcal{B}$, they will use this iteration to contribute towards their raw key in the obvious way.
\end{enumerate}

Note that, when $\alpha = 1/\sqrt{2}$, this protocol is exactly the three-state version of BB84 discussed in \cite{QKD-BB84-three-state,QKD-BB84-three-state-v2}; that is, it is BB84 with the limitation that $A$ can never send $\ket{-}$ and thus the users can never measure the probability of $E$'s attack flipping a $\ket{-}$ to a $\ket{+}$ (they can only measure the probability of a $\ket{+}$ flipping to a $\ket{-}$).  As discussed in \cite{QKD-BB84-three-state}, there are several potential practical benefits to this protocol.  It is also interesting theoretically as it allows us to study the effects of decreasing $A$'s required quantum capabilities.

\section{Security Proof}

To compute a lower bound on the key rate $r$, we must first describe the quantum system after one iteration of the protocol, conditioning on the event this iteration is used to contribute towards the raw key (in particular, $A$ sends a state from $\mathcal{B}$ and $B$ measures in that same basis).  We will first assume collective attacks; later we will comment on general attacks.

Fix $\alpha \in (0,1)$; this parameter is public knowledge (in particular $E$ also knows the value this is set to) and furthermore, once fixed it is constant throughout the protocol.  Let $U$ be the unitary attack operator $E$ employs each iteration of the protocol.  Since we are conditioning on the event this iteration is used to contribute to the raw key, $A$ will prepare $\ket{0}$ or $\ket{1}$, choosing each with probability $1/2$.  $E$ will then attack with operator $U$.  Without loss of generality, we may assume $E$'s ancilla is cleared to some $\ket{0}_E$ state and so write $U$'s action on basis states as follows:
\begin{align}
U\ket{0,0} = \ket{0,e_0} + \ket{1,e_1} \text{ } , \text{ } U\ket{1,0} = \ket{0,e_2} + \ket{1,e_3}\label{eq:attack-op}
\end{align}
These $\ket{e_i}$ are arbitrary states in $E$'s ancilla and are not assumed to be normalized nor orthogonal.  Unitarity of $U$ of course imposes certain conditions on them; furthermore, parameter estimation will yield even more data on them later.

Thus, the state of the quantum system, when the qubit arrives at $B$'s lab is: $\frac{1}{2}\kb{0}_A \otimes P(\ket{0,e_0} + \ket{1,e_1}) + \frac{1}{2}\kb{1}_A\otimes P(\ket{0,e_2} + \ket{1,e_3})$, where $P(z) = zz^*$ and $z^*$ is the conjugate transpose of $z$.

$B$ will then measure in the $\mathcal{B}$ basis (again, we are conditioning on the event this iteration is used for the raw key) yielding the final quantum state:
\begin{align*}
&\frac{1}{2}\kb{00}_{AB} \otimes \kb{e_0} + \frac{1}{2}\kb{11}_{AB} \otimes \kb{e_3}\\
+&\frac{1}{2}\kb{01}_{AB} \otimes \kb{e_1} + \frac{1}{2}\kb{10}_{AB} \otimes \kb{e_2}\\
\end{align*}
Call this state $\chi_{ABE}$.  We will make the usual assumption that the noise in the $\mathcal{B}$ basis, induced by $E$'s attack, is symmetric: that is $\braket{e_0|e_0} = \braket{e_3|e_3} = 1-Q$ and $\braket{e_1|e_1} = \braket{e_2|e_2} = Q$.  These parameters $\braket{e_i|e_i}$ can obviously be estimated by $A$ and $B$; thus this symmetry condition can even be enforced.  Note that, without this assumption, our analysis could still be carried out, though the algebra is not as amiable.

Assume, for the moment, that $Q > 0$ (we will consider the case $Q=0$ later).  Tracing out $B$, we may write $\chi_{ABE}$ as:
\begin{equation}\label{eq:final-state-chi}
\chi_{AE} = (1-Q)\rho_{AE} + Q\sigma_{AE},
\end{equation}
where:
\begin{align}
\rho_{AE} &= \frac{1}{2}\kb{0}_{A}\otimes\frac{\kb{e_0}}{1-Q} + \frac{1}{2}\kb{1}_A\otimes\frac{\kb{e_3}}{1-Q}\\
\sigma_{AE} &= \frac{1}{2}\kb{0}_A\otimes\frac{\kb{e_1}}{Q} + \frac{1}{2}\kb{1}_A\otimes\frac{\kb{e_2}}{Q}.
\end{align}
It is obvious that both $\rho_{AE}$ and $\sigma_{AE}$ are Hermitian semi-definite operators of unit trace.  Before continuing, we require two lemmas which, though trivial to prove, we include for completeness:
\begin{lemma}\label{lemma1}
Given a finite dimensional Hilbert space $\mathcal{H} = \mathcal{H}_C\otimes\mathcal{H}_E$, let $\{\ket{1}_C, \cdots, \ket{n}_C\}$ be an orthonormal basis of $\mathcal{H}_C$.  Consider the following density operator:
\[
\rho = \sum_{i=1}^n p_i \kb{i}_C \otimes \sigma_E^{(i)},
\]
where $\sum p_i = 1$, $p_i \ge 0$, and each $\sigma_E^{(i)}$ is a Hermitian positive semi-definite operator of unit trace acting on $\mathcal{H}_E$.  Then:
\[
S(\rho) = H(p_1, \cdots, p_n) + \sum_{i=1}^n p_iS\left(\sigma_E^{(i)}\right).
\]
\end{lemma}
\begin{proof}
See \cite{SQKD-Krawec-SecurityProof} for a proof.
\end{proof}

%%%%
\begin{lemma}\label{lemma2}
Given a finite dimensional Hilbert space $\mathcal{H} = \mathcal{H}_A\otimes\mathcal{H}_E$ and the following density operators:
\begin{align*}
\rho_{AE} &= p_0 \kb{0}_A \otimes \rho_E^0 + p_1\kb{1}_A \otimes \rho_E^1\\
\sigma_{AE}&= q_0\kb{0}_A \otimes \sigma_E^0 + q_1\kb{1}_A \otimes \sigma_E^1\\\\
\chi_{AE} &= p_\rho\rho_{AE} + p_{\sigma}\sigma_{AE},
\end{align*}
then the following is true:
\begin{align}
&S(A|E)_\chi \ge p_\rho S(A|E)_\rho + p_\sigma S(A|E)_\sigma\label{eq:thm1-1}
\end{align}
\end{lemma}
\begin{proof}
Let $\mathcal{H}_C$ be the two dimensional Hilbert space spanned by orthonormal basis $\{\ket{X}, \ket{Y}\}$ and let $\chi_{AEC}$ be the following density operator:
\[
\chi_{AEC} = p_\rho \kb{X} \otimes \rho_{AE} + p_\sigma \kb{Y} \otimes \sigma_{AE},
\]
which acts on $\mathcal{H}_A\otimes\mathcal{H}_E\otimes\mathcal{H}_C$.  Observe that $tr_C\chi_{AEC} = \chi_{AE}$.  Due to the strong sub additivity of von Neumann entropy, it holds that: $S(A|E)_\chi \ge S(A|EC)_{\chi}$.  We will show that:
\begin{equation}\label{eq:thm1-3}
S(A|EC)_{\chi} = p_\rho S(A|E)_\rho + p_\sigma S(A|E)_\sigma,
\end{equation}
from which Equation \ref{eq:thm1-1} will follow.  Of course $S(A|EC)_\chi = S(AEC)_\chi - S(EC)_\chi$.  Applying Lemma \ref{lemma1} twice, we have:
\begin{align*}
S(AEC)_\chi &= h(p_\rho) + p_\rho S(AE)_\rho + p_\sigma S(AE)_\sigma\\
S(EC)_\chi &= h(p_\rho) + p_\rho S(E)_\rho + p_\sigma S(E)_\sigma.
\end{align*}

Thus:
\begin{align*}
S(A|EC)_\chi &= S(AEC)_\chi - S(EC)_\chi\\
&= p_\rho (S(AE)_\rho - S(E)_\rho)\\
&+ p_\sigma (S(AE)_\sigma - S(E)_\sigma)\\
&= p_\rho S(A|E)_\rho + p_\sigma S(A|E)_\sigma.
\end{align*}

\end{proof}
%%%%

Thus, from Lemma \ref{lemma2}, we may compute a lower bound on $S(A|E)_\chi$ (and thus a lower bound on the key rate $r$).  This is:
\[
S(A|E)_\chi \ge (1-Q)\cdot S(A|E)_\rho + Q\cdot S(A|E)_\sigma.
\]

We need now only compute the conditional entropy of the operators $\rho_{AE}$ and $\sigma_{AE}$ individually.  It is clear that:
\[
S(AE)_\rho = S(AE)_\sigma = h(1/2) = 1.
\]

Now, tracing out $A$, yields:
\begin{align}
\rho_E &= \frac{1}{2(1-Q)}(\kb{e_0} + \kb{e_3})\\
\sigma_E &= \frac{1}{2Q}(\kb{e_1} + \kb{e_2}).
\end{align}

We must now compute $S(E)_\rho$ and $S(E)_\sigma$.  To do so, we will need the eigenvalues of these two density operators.  First consider $\rho_E$.  Without loss of generality, we may write:
\begin{align}
\ket{e_0} = z\ket{E} \text{ }, \text{ } \ket{e_3} = he^{i\theta}\ket{E} + d\ket{I},
\end{align}
where $z, h, d \in \mathbb{R}$, $\braket{E|E} = \braket{I|I} = 1$, and $\braket{E|I} = 0$.  Furthermore, we have the following:
\begin{align}
&z^2 = h^2 + d^2 = 1-Q\label{eq:basis-ident-1}\\
&hze^{i\theta} = \braket{e_0|e_3} \Rightarrow h^2z^2 = |\braket{e_0|e_3}|^2.\label{eq:basis-ident-2}
\end{align}
In this $\{\ket{E}, \ket{I}\}$ basis, we may write $\rho_E$ as:
\[
\rho_E = \frac{1}{2(1-Q)} \left( \begin{array}{cc}
z^2 + h^2 & he^{i\theta}d\\\\
he^{-i\theta}d & d^2\end{array}\right),
\]
the eigenvalues of which are readily computed to be:
\[
\lambda_\pm^\rho = \frac{1}{2} \pm\frac{ \sqrt{(z^2+h^2-d^2)^2+4h^2d^2} }{4(1-Q)}
\]
Using identities \ref{eq:basis-ident-1} and \ref{eq:basis-ident-2}, we have:
\begin{align}
\lambda_\pm^\rho &= \frac{1}{2}\pm\frac{ \sqrt{4h^4 + 4h^2(z^2-h^2)} }{4(1-Q)}\notag\\
&=\frac{1}{2} \pm \frac{ \sqrt{h^2z^2} }{2(1-Q)} =\frac{1}{2}\pm\frac{|\braket{e_0|e_3}|}{2(1-Q)}.
\end{align}

Similarly, we may compute the eigenvalues of the two-dimensional operator $\sigma_E$ as:
\begin{equation}
\lambda_\pm^\sigma = \frac{1}{2} \pm \frac{|\braket{e_1|e_2}|}{2Q}.
\end{equation}
(We are still assuming, for now, $Q > 0$.)

Thus, $S(E)_\rho = h(\lambda_+^\rho)$ and $S(E)_\sigma = h(\lambda_+^\sigma)$, and so:
\begin{align*}
S(A|E)_\chi &\ge (1-Q)(1 - h(\lambda_+^\rho)) + Q(1 - h(\lambda_+^\sigma))\\
&\ge 1 - (1-Q)\cdot h(\lambda_+^\rho) - Q\cdot h(\lambda_+^\sigma).
\end{align*}

It is trivial to show that $H(A|B) = h(Q)$.  Thus, the key rate of this three-state protocol is lower bounded by:
\begin{align*}
r &= S(A|E)_\chi - H(A|B)\\
&\ge 1 - (1-Q)\cdot h(\lambda_+^\rho) - Q\cdot h(\lambda_+^\sigma) - h(Q).
\end{align*}

Note that, if $Q = 0$, then $\ket{e_1} \equiv \ket{e_2} \equiv 0$ and so these terms never show up in Equation \ref{eq:final-state-chi}.  In this case, we may define $\lambda_+^\sigma$ arbitrarily and the above key rate bound will still hold.

Therefore, to determine the key rate, $A$ and $B$ must estimate the quantities $|\braket{e_0|e_3}|$ and $|\braket{e_1|e_2}|$.  These cannot be directly observed; however, by using the error rate in the $\mathcal{A}$ basis, along with mismatched measurement results, they may determine bounds on these two quantities.  Before discussing how this is done, however, we mention one last critical detail.  Note that $\lambda_+^\rho$ and $\lambda_+^\sigma$ are both greater than, or equal to, $1/2$.  Note also that the binary entropy function $h(x)$ attains its maximum when $x = 1/2$ and on the interval $[1/2, 1]$ it is a decreasing function.  Thus, if we find values $\lambda^\rho$ and $\lambda^\sigma$, such that $1/2 \le \lambda^\rho \le \lambda_+^\rho$ and $1/2 \le \lambda^\sigma \le \lambda_+^\sigma$, then it will hold that $h(\lambda_+^\rho) \le h(\lambda^\rho)$ (and similarly for $\sigma$).  Thus, we have:
\begin{align}
r &\ge 1 - (1-Q)\cdot h(\lambda_+^\rho) - Q\cdot h(\lambda_+^\sigma) - h(Q)\notag\\
&\ge 1 - (1-Q)\cdot h(\lambda^\rho) - Q\cdot h(\lambda^\sigma) - h(Q).\label{eq:keyrate-bound-final}
\end{align}

We will soon see that it is easier to bound the real parts of $\braket{e_0|e_3}$ and $\braket{e_1|e_2}$.  Therefore, we will define:
\begin{align*}
\lambda^\rho = \frac{1}{2} + \frac{|Re\braket{e_0|e_3}|}{2(1-Q)} \ge \frac{1}{2} \text{ } , \text{ } \lambda^\sigma &= \frac{1}{2} + \frac{|Re\braket{e_1|e_2}|}{2Q} \ge \frac{1}{2}
\end{align*}
(the same discussion before concerning the case when $Q=0$ applies).  It is clear that:
\begin{align*}
\lambda_+^\rho &= \frac{1}{2}+\frac{|\braket{e_0|e_3}|}{2(1-Q)} = \frac{1}{2}+\frac{\sqrt{Re^2\braket{e_0|e_3} + Im^2\braket{e_0|e_3}}}{2(1-Q)}\\
&\ge \frac{1}{2}+\frac{\sqrt{Re^2\braket{e_0|e_3}}}{2(1-Q)} = \lambda^\rho \ge \frac{1}{2},
\end{align*}
and similarly, $\lambda_+^\sigma \ge \lambda^\sigma \ge 1/2$.  To evaluate our key rate bound in Equation \ref{eq:keyrate-bound-final}, we therefore need to determine bounds on the real part only of $\braket{e_0|e_3}$ and $\braket{e_1|e_2}$.

\subsection{Parameter Estimation}

To estimate $Re\braket{e_0|e_3}$ and $Re\braket{e_1|e_2}$, we will consider, first, the noise in the $\mathcal{A}$ basis.  To improve this bound, we will also consider measurement results from mismatched bases (results which are typically discarded).  Before, continuing, let us introduce some notation.  In particular, we will denote by $\re{i}{j}$ to mean $Re\braket{e_i|e_j}$ (for $i,j \in \{0,1,2,3\}$.  By $p_{x,y}$ we mean the probability that if $A$ sends $\ket{x}$ (for $x \in \{0,1,a\}$), then, after $E$'s attack, $B$ measures $\ket{y}$ (for $y \in \{0,1,a,\bar{a}\}$).  For instance, $p_{0,1} = \braket{e_1|e_1} = Q$.  Finally, let $\beta = \sqrt{1-\alpha^2}$ (and so $\ket{a} = \alpha\ket{0} + \beta\ket{1}$).

Now, consider the quantity $p_{a,\bar{a}}$ which, to avoid confusion with $p_{a,a}$, we will also denote $Q_\mathcal{A}$.  This represents the error in the $\mathcal{A}$ basis (note that, unlike in a four-state protocol, $A$ and $B$ cannot directly measure $p_{\bar{a}, a}$).  By linearity of $U$, we have (see Equation \ref{eq:attack-op}):
\begin{align}
U\ket{a} &= \ket{0}(\alpha\ket{e_0} + \beta\ket{e_2}) + \ket{1}(\alpha\ket{e_1} + \beta\ket{e_3})\label{eq:PE-Ua}\\\notag\\
&=\ket{a}(\alpha^2\ket{e_0} + \alpha\beta\ket{e_2} + \alpha\beta\ket{e_1} + \beta^2\ket{e_3})\notag\\
&+\ket{\bar{a}}(\beta\alpha\ket{e_0} + \beta^2\ket{e_2} - \alpha^2\ket{e_1} - \alpha\beta\ket{e_3}).\notag
\end{align}
Thus:
\begin{align}
Q_\mathcal{A} &= \alpha^2\beta^2(\braket{e_0|e_0} + \braket{e_3|e_3}) + \beta^4\braket{e_2|e_2} + \alpha^4\braket{e_1|e_1}\notag\\
&+2Re(\beta^3\alpha\braket{e_0|e_2} - \beta\alpha^3\braket{e_0|e_1} - \alpha^2\beta^2\braket{e_0|e_3}\notag\\
&-\alpha^2\beta^2\braket{e_1|e_2} - \alpha\beta^3\braket{e_2|e_3} + \alpha^3\beta\braket{e_1|e_3}).\label{eq:PE-A-noise}
\end{align}

Of course, $\alpha$ (and thus $\beta$) are public knowledge; also the quantities $\braket{e_i|e_i}$ can be estimated using the $\mathcal{B}$-basis noise.  That leaves: $\re{0}{2}, \re{0}{1}, \re{0}{3}, \re{1}{2}, \re{2}{3}, \re{1}{3}$.  Most of these, however, may be estimated using mismatched measurement results.  Consider the quantity $p_{0,a}$.  This is a value that would ordinarily be discarded as $A$ and $B$ used different bases; yet, this probability, which can be estimated by $A$ and $B$, will lead to an estimate of $\re{0}{1}$.  Indeed:
\begin{align*}
U\ket{0} &= \ket{0,e_0} + \ket{1,e_1}\\
&= \ket{a}(\alpha\ket{e_0} + \beta\ket{e_1}) + \ket{\bar{a}}(\beta{e_0} - \alpha\ket{e_1}),
\end{align*}
and so:
\begin{align}
&p_{0,a} = \alpha^2 \braket{e_0|e_0} + \beta^2\braket{e_1|e_1} + 2\alpha\beta\re{0}{1}\notag\\
\Rightarrow& \re{0}{1} = \frac{p_{0,a} - \alpha^2(1-Q) - \beta^2Q}{2\alpha\beta}.\label{eq:PE-01}
\end{align}
On observing $Q$, along with $p_{0,a}$, $A$ and $B$ immediately may estimate $\re{0}{1}$.  Similarly, they may use $p_{1,a}$ to estimate $\re{2}{3}$:
\begin{equation}
\re{2}{3} = \frac{p_{1,a} - \alpha^2Q - \beta^2(1-Q)}{2\alpha\beta}.
\end{equation}

When $\alpha = 1/\sqrt{2}$ (as dictated by the original three-state protocol \cite{QKD-BB84-three-state,QKD-BB84-three-state-v2}), $A$ and $B$ need not estimate $\re{0}{2}$ and $\re{1}{3}$ due to the fact that unitarity of $U$ (which imposes the condition $\re{0}{2} = -\re{1}{3}$) will force the terms to cancel in Equation \ref{eq:PE-A-noise}.  For other values of $\alpha$, however, an estimate will be necessary and it may be accomplished by considering $p_{a,0}$.  From Equation \ref{eq:PE-Ua} this is:
\begin{align}
&p_{a,0} = \alpha^2\braket{e_0|e_0} + \beta^2\braket{e_2|e_2} + 2\alpha\beta\re{0}{2}\notag\\
\Rightarrow& \re{0}{2} = \frac{p_{a,0} - \alpha^2(1-Q) - \beta^2Q}{2\alpha\beta}.\label{eq:PE-02}
\end{align}
Since unitarity of $U$ forces the relation $\re{1}{3} = -\re{0}{2}$, $A$ and $B$ now have estimates of all quantities in Equation \ref{eq:PE-A-noise} except for $\re{0}{3}$ and $\re{1}{2}$.  However, the Cauchy-Schwarz inequality forces $\re{1}{2} \in [-Q,Q]$.  This, combined with Equation \ref{eq:PE-A-noise}, allow the users to bound $\re{0}{3}$ and thus evaluate $r$ (one must simply find the minimum $r$ over all $\re{1}{2} \in [-Q,Q]$).  For our evaluations, we performed this optimization numerically; however finding an analytic solution would be straight-forward.  Finally, as the protocol is permutation invariant, the results of \cite{QKD-general-attack,QKD-general-attack2} apply and thus our key rate bound holds even against the most general of attacks (not only collective).

\subsection{Evaluation}

\begin{figure}
  \centering
 \includegraphics[width=200pt]{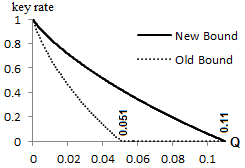}
\caption{Comparing our new key rate bound (for any $\alpha \in (0,1)$) with the one from \cite{QKD-BB84-three-state-v2} (which did not use mismatched measurement outcomes).}\label{fig:symmetric}
\end{figure}

To evaluate our key rate bound, we will first consider the case that $E$'s attack is symmetric in that it may be modeled as a depolarization channel (this is a common assumption in QKD protocol security proofs; in fact, it could even be enforced by the users).  Consider a depolarization channel with parameter $Q$: $\mathcal{E}_Q(\rho) = (1-2Q)\rho + QI$ where $I$ is the two-dimensional identity operator.  Then, if $A$ sends $\ket{i} \in \mathcal{B}$, the probability of measuring $\ket{1-i}$ is $Q$ as desired.  Thus $\braket{e_1|e_1} = \braket{e_2|e_2} = Q$.  Furthermore, if $A$ sends $\ket{0}$, then the qubit's state when it arrives at $B$'s lab is:
\[
\mathcal{E}_Q(\kb{0}) = (1-2Q)\kb{0} + Q(\kb{0} + \kb{1}),
\]
and so we have $p_{0,a} = (1-2Q)\alpha^2 + Q$ (note that if $\alpha = 1/\sqrt{2}$ and thus $\ket{a} = \ket{+}$, we have $p_{0,+} = 1/2$ as expected).  Using this value in Equation \ref{eq:PE-01} yields:
\begin{align*}
\re{0}{1} &= \frac{(1-2Q)\alpha^2+Q - \alpha^2(1-Q) - \beta^2Q}{2\alpha\beta}\\
&=\frac{\alpha^2-2\alpha^2Q+Q-\alpha^2+\alpha^2Q-(1-\alpha^2)Q}{2\alpha\beta} = 0.
\end{align*}
Similarly, we find $\re{2}{3} = 0$.

If $\alpha = 1/\sqrt{2}$, we are done; otherwise, we must consider $p_{a,0}$.  If $A$ sends $\ket{a}$, then the qubit arriving at $B$'s lab is:
\begin{equation}
\mathcal{E}_Q(\kb{a}) = (1-2Q)\kb{a} + Q(\kb{a}+\kb{\bar{a}}),\label{eq:eval-pw}
\end{equation}
from which it is clear that $p_{a,0} = (1-2Q)\alpha^2 + Q$.  Substituting into Equation \ref{eq:PE-02} we find that $\re{0}{2} = 0$.  Thus also $\re{1}{3} = -\re{0}{2} = 0$.

Substituting this into Equation \ref{eq:PE-A-noise} and solving for $\re{0}{3}$ yields the expression:
\begin{align}
\re{0}{3} &= \frac{2\alpha^2\beta^2(1-Q) + (\beta^4+\alpha^4)Q - Q_\mathcal{A} - 2\alpha^2\beta^2\re{1}{2}}{2\alpha^2\beta^2}\notag\\
&=1 - 2Q + \frac{Q - Q_\mathcal{A}}{2\alpha^2\beta^2} - \re{1}{2},
\end{align}
where, above, we used the fact that: $1 = (\alpha^2 + \beta^2)^2 = \alpha^4 + 2\alpha^2\beta^2 + \beta^4 \Rightarrow \alpha^4 + \beta^4 = 1 - 2\alpha^2\beta^2$.

Finally, from Equation \ref{eq:eval-pw}, we find $Q_\mathcal{A} = Q$ and so: $\re{0}{3} = 1 - 2Q - \re{1}{2}$.  Note that, under this (entirely enforceable) assumption of a depolarization channel, the parameters $\alpha$ and $\beta$ do not show up in the above expression; they therefore do not appear in the evaluation of $r$.  Thus, in this symmetric case, the value of $\alpha$ is irrelevant (at least in the perfect qubit, asymptotic scenario - we make no claims about its relevance in more practical scenarios; it also is relevant in non-symmetric scenarios as we soon show).

To evaluate the key rate $r$, we must now simply minimize $r$ over all $\re{1}{2} \in [-Q,Q]$.  We performed this computation numerically, resulting in the key rate shown in Figure \ref{fig:symmetric}.  Notice that the key rate is positive for all $Q \le 11\%$; this is exactly the tolerance supported by the four state BB84 protocol.  Compare this with the lower-bound from \cite{QKD-BB84-three-state-v2} which did not make use of mismatched measurement bases and which only remained positive for $Q \le 5.1\%$.  In fact, we found that, in this symmetric case, our new key rate bound agrees exactly with that of the four-state BB84 protocol $1-h(Q) - h(Q_X)=1-2h(Q)$ where $Q_X$ is the error in the $X$ ($\{\ket{\pm}\}$) basis \cite{QKD-renner-keyrate}.  Furthermore, this is true for any $\alpha \in (0,1)$.  This provides an alternative proof to the result from \cite{QKD-Tom-threestate1}.

While $\alpha$ does not appear when $E$'s attack is symmetric; this of course is not true in other cases.  We considered the effect of various settings for $\alpha$ when $Q_\mathcal{A} = 2Q$ (Figure \ref{fig:fig2} (a)) and $Q_\mathcal{A} = Q/2$ (Figure \ref{fig:fig2} (b)); we also plotted the old key rate bound from \cite{QKD-BB84-three-state-v2} for comparison.  In Figure \ref{fig:fig2} (c) and (d) we considered the case when $p_{a,0}$ and $p_{1,a}$ are such that $\re{0}{2}$ and $\re{2}{3}$ are non-zero.  As $\alpha$ is a parameter that must be set before the protocol runs, it seems $\alpha=1/\sqrt{2}$ is a good compromise (other settings can do better or worse depending on the attack used).

\begin{figure}
  \subfloat[$Q_\mathcal{A} = 2Q$]{\includegraphics[width=115pt]{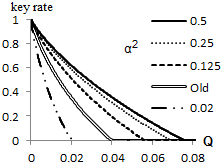}} \qquad
  \subfloat[$Q_\mathcal{A} = Q/2$]{\includegraphics[width=115pt]{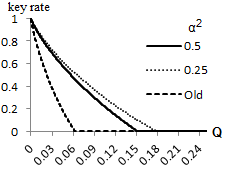}} \qquad
  \subfloat[$\re{0}{2} = -\sqrt{Q(1-Q)}$]{\includegraphics[width=115pt]{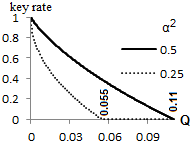}} \qquad
  \subfloat[$\re{2}{3} = \sqrt{Q(1-Q)}$]{\includegraphics[width=115pt]{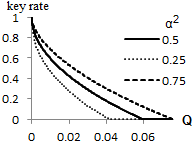}} \qquad
\caption{Showing how $\alpha^2$ affects the key rate in various noise scenarios.  Also comparing with the ``Old'' bound from \cite{QKD-BB84-three-state-v2} in (a) and (b).  (a): when $Q_\mathcal{A} = 2Q$. (b): when the $\mathcal{A}$ noise is half. (c): When the noise in both bases are equal, but $p_{a,0}$ is such that $\re{0}{2} = -\sqrt{Q(1-Q)}$ (the largest negative it could be). (d): Like (c), but now $\re{2}{3} = \sqrt{Q(1-Q)}$ (while $\re{0}{2} = 0$).}\label{fig:fig2}
\end{figure}

\section{Closing Remarks}

We have derived a new proof of security and key-rate bound for a three state BB84 protocol - a protocol where $A$ sends only $\ket{0}$, $\ket{1}$, or $\ket{a} = \alpha\ket{0} + \sqrt{1-\alpha^2}\ket{1}$.  Furthermore we have shown that this new key rate bound, in addition to the use of mismatched measurement outcomes, can tolerate the same maximal noise level as the four state BB84 (i.e., the key rate is positive for all $Q \le 11\%$) in the asymptotic scenario.  The technique we used in our proof - especially our use of mismatched measurement outcomes in this manner to estimate $\re{i}{j}$ despite the lack of the statistic $p_{\bar{a},a}$ - may hold application in the analysis of other QKD protocols where one (or both) party is limited.  This is a subject we intend to investigate in the near future.  %Also, it would be interesting (and important) to consider the finite key setting.\newline

\end{document}